\documentclass[12pt]{article}
\usepackage[margin=1.3in]{geometry}
\usepackage{amsthm}
\usepackage{amsmath}
\usepackage{natbib}
\usepackage[colorlinks,citecolor=blue,urlcolor=blue,filecolor=blue,backref=page]{hyperref}
\usepackage{graphicx}
\usepackage{color}
\usepackage{bm}
\usepackage{bbold}
\usepackage{authblk}
\usepackage{colortbl}
\definecolor{lightgray}{gray}{0.9}
\definecolor{darkgray}{gray}{0.7}
\usepackage{amsthm}

\newcommand{\argmax}{\operatornamewithlimits{argmax}}
\newcommand{\argmin}{\operatornamewithlimits{argmin}}
\newcommand{\iidsim}{\overset{iid}{\sim}}
\let\proglang=\textsf
\newcommand{\pkg}[1]{{\fontseries{b}\selectfont #1}}
\newcommand*{\defeq}{\mathrel{\vcenter{\baselineskip0.5ex              \lineskiplimit0pt
    \hbox{\footnotesize.}\hbox{\footnotesize.}}}%
     =}
\newcommand{\n}[1]{\left\lVert#1\right\rVert^2}
\newtheorem{theo}{Theorem}
\newtheorem{mydef}{Definition}
\newtheorem{myprop}{Proposition}
\providecommand{\keywords}[1]{\textbf{\textit{Keywords ---}} #1}
\definecolor{DarkGreen}{rgb}{0,0.6,0.5}
\definecolor{Purple}{rgb}{0.3,0,0.6}

\title{Informative extended Mallows priors in the Bayesian Mallows model}
\author[1]{Marta Crispino}
\author[2]{Isadora Antoniano-Villalobos}
\affil[1]{Mistis team, Inria Grenoble Rh\^{o}ne-Alpes, Inovall\'ee 655 avenue de l'Europe, Montbonnot France.}
\affil[2]{Department of Environmental Sciences, Informatics and Statistics, Ca' Foscari University of Venice, Venice, Italy and Bocconi Institute for Data Science and Analytics, Bocconi University, Milan, Italy.}
\date{}
\begin{document}
\maketitle

\begin{abstract}
The aim of this work is to study the problem of prior elicitation for the Mallows model with Spearman's distance, a popular distance-based model for rankings or permutation data. Previous Bayesian inference for such model has been limited to the use of  the uniform prior over the space of permutations. We present a novel strategy to elicit subjective prior beliefs on the location parameter of the model, discussing the interpretation of hyper-parameters and the implication of prior choices for the posterior analysis. 
\end{abstract}


\section{Motivation}\label{sec:moti}

In recent years, interest in preference data has increased, in part due to internet-related activities. The study of rankings, in particular, has received special attention, since this type of data arise in many fields. Notable examples are electoral systems in which voters are required to rank candidates, as is the case of the Irish general elections \citep{gormley2008mixture}; automatic recommender systems seeking to aggregate preferences in order to suggest products to the customers \citep{SunLebanonKidwell2012}; market research based on surveys in which competing services, or items, are compared or ranked by customers \citep{carconf}; medical applications, specially in genomics, in which genes are sometimes ranked according to their expression levels under various experimental conditions \citep{vitelli17}, and other data is often transformed into rankings in order minimize the effect of miscalibration error from the measuring devices
\citep{mollica14}.
In a coherent analysis of ranking data, the quantification of uncertainty regarding the estimated quantities is a fundamental aspect of decision making. It could, for instance, allow for actions grounded on unreliable estimates to be deferred until more data are available. 

The Mallows model (MM) \citep{Mallows1957, Diaconis1988} is a popular two-parameter distance-based family of models for ranking data, based on the assumption that a modal ranking, which can be interpreted as the consensus ranking of the population, exists. The probability of observing a given ranking is then assumed to decay exponentially fast as its distance from the consensus grows. Individual models with different properties can be obtained depending on the choice of distance on the space of permutations. The scale or precision parameter, controlling the concentration of the distribution, determines the rate of decay of the probability of individual ranks. 

We focus on the Mallows model with Spearman's distance (MMS), introduced by \citep{Mallows1957} with the name of rho-model, since Spearman's distance, when re-scaled to lie between $-1$ and $1$, arises naturally as the correlation between the ranks of two samples. \citet{Fligner1990} and \citet{vitelli17} have studied Bayesian inference for the MMS, limiting the analysis to the use of a uniform prior on the consensus ranking. As we discuss in Section \ref{sub:thetaknown}, this can be interpreted as a non-informative prior.  

Within the Bayesian literature, non-informative and objective priors have attracted much attention in the search for standard go-to procedures when prior information is unavailable. They can also be used to provide a sense of neutrality to the analysis by allowing the data to be the only source of information in the estimation procedure. However, when information is available from experts or external sources, it may be argued that a fully Bayesian analysis should include this subjective prior belief. \cite{dawid1997comments} clearly stated that ``no theory which incorporates non-subjective priors can truly be called Bayesian, and no amount of wishful thinking can alter this reality''. While admitting that both approaches may be valid in different situations, in this paper we explore the possibility of including genuine prior information, which might come from a literature review,  from an expert or from an earlier data analysis, into the Bayesian Mallows model for ranking data \citep{Mallows1957,vitelli17}. 

Previous proposals to include prior information on the consensus ranking of a MM include \cite{damien2012}, who suggest eliciting a prior on the consensus which is constant on conjugacy classes. In other words, they propose a prior that assigns \emph{a priori} equal probability to all permutations with the same cyclic structure. However, the conjugate classes defined by cyclic structures do not coincide with those defined by permutations lying at the same distance  (e.g. Spearman's) from the consensus ranking, making this approach impractical for the MMS, as it is difficult to assess a way in which prior information enters the model. \citet{MeilaBao2010} and \citet{Meila2010} consider the MM with Kendall's distance within the Bayesian paradigm and provide a conjugate prior for the model parameters which is known up to a normalization constant. However, their analysis does not extend to the MMS.
\cite{XuAlvoYu2018} propose an alternative family of models for rankings, based on a mapping of the data to the unit sphere \citep[see also][]{McCullagh1993}. The location parameter of their model has an interpretation analogous to that of the consensus ranking but it is not limited to be itself a ranking, thus allowing to express a more general form of consensus. The MMS is a particular case of this model, and the authors propose a conjugate Bayesian prior for the consensus parameter. However, the emphasis of the paper is on efficient inference via an approximation of the model's normalizing constant and the use of variational methods; prior elicitation and the inclusion of prior information are not discussed.

In the present work, which stems from Chapter 6 of \citet{phdcrispino}, we aim to provide experts using the MMS with a tool to express their beliefs, knowing the effect of prior choices in their analysis, should they wish to do so.  
With this in mind, by exploiting the notion of permutohedron, also known as permutation polytope, \citep{thompson,McCullagh1993, Marden1995}, we find an explicit form for a conjugate prior on the consensus parameter for the MMS. We then study its properties, presenting some theoretical insights on the prior elicitation problem.  Subjective prior information on the consensus ranking can therefore be elicited  by choosing proper hyper-parameters. In doing this, we initially assume the scale parameter of the MMS to be known, given that in most applications it is considered a nuisance, the interest being focused on the estimation of the consensus ranking \citep[see][Section 3]{vitelli17}.
The proposed prior density can handle a situation when only partial information is available, which is particularly relevant when the set of items to be ranked is very large. In such cases it is unlikely that a full ranking is \emph{a priori} available, while it could be possible to express some prior belief regarding which are the most (or least) preferred items. An additional advantage of our prior is given by the interpretability of the hyper-parameters in terms of the amount and type of information included.

The paper is organized as follows.
In Section \ref{sec:preli} we give an overview of the MMS.
In Section \ref{sec:prior} we discuss 
the novel results regarding the conjugate prior for the consensus parameter of the MMS, initially assuming the dispersion parameter to be known (Section \ref{sub:thetaknown}), then (Section \ref{sub:thetaunknown}) working with both parameters unknown. 
In Section \ref{sec:inference} we sketch the MCMC algorithm used to perform inference on our model, and in Section \ref{sec:example} we illustrate the inference on simple examples, exploiting both simulations and benchmark datasets. We conclude with some final remarks in Section \ref{sec:conclusion}. 

\section{Preliminaries}\label{sec:preli}
A (full) ranking of $n$ items, or $n$-ranking is defined as a map from a finite set, $\{A_1,...,A_n\}$, of labeled items to the space $\mathcal{P}_n$ of $n$-dimensional permutations. A ranking can, therefore, be represented by a vector $\bm{r} = (r_1,\ldots,r_n)$, where $r_i$ is the rank assigned to item $A_i$ according to some criterion. Formally, individual ranks are ordinal numbers, so that $r_i < r_j$ when item $A_i$ is preferred to (ranked lower than) item $A_j$.  Alternatively, rank data may be represented through orderings, which are ordered vectors of labels. Clearly, there is a one-to-one relationship between the two representations, e.g. a possible ranking of the set $A_1,\ldots,A_{5}$ is $\bm{r} = (1, 3, 4, 5, 2)$, corresponding to the ordering $\bm{o} = (A_1, A_5, A_2, A_3, A_4)$. Since the ranking vector representation has many advantages in terms of modelling, we will stick to it throughout the paper, and only use the orderings when necessary for illustrative purposes. 
Given the trivial one-to-one relation between ordinal and cardinal numbers, with a slight abuse of notation, one may consider $n$-rankings as $n$-dimensional vectors obtained by permuting the first natural numbers, $\{1,\ldots,n\}$. It is then easy to see that $\mathcal{P}_n$ is contained in a  $(n-1)$-dimensional affine  subspace  of $\mathbb{R}^n$. In fact, it is composed by the $n!$ points on the intersection between the hyper-plane with coordinate sums equal to $s_n=n(n+1)/2$ and the surface of an $n$-dimensional sphere of squared radius $c_n=n(n+ 1)(2n+ 1)/6$ centered at the origin.  Thus,  all the  points  of $\mathcal{P}_n$ lie  on  an  $(n-1)$-dimensional  sphere  of  squared  radius $\phi_n=n(n^2-1)/12$ centered at $\frac{(n+1)}{2}\bm 1_n$, where $\bm 1_n\in \mathbb{R}^n$ denotes the vector with all entries equal to $1$ \citep{McCullagh1993}. 

The Mallows model for ranking data \citep{Mallows1957} defines the probability that a random $n$-ranking  $\bm{R}$ takes a  value $\bm r\in\mathcal{P}_n$ as 
 \begin{equation}\label{eq:Mallows}
  \mathbb{P}(\bm{R}=\bm{r}\,|\bm{\rho},\theta,d) = \frac{1}{Z_d(\theta)}\exp\left[-\theta\, d(\bm{r}, \bm{\rho})\right],
 \end{equation}
where $\bm{\rho}\!\in\!\mathcal{P}_n$ is a location parameter representing the shared consensus ranking and $\theta\!\geq\!0$ is a scale parameter describing the concentration of the mass around the shared consensus. Different families of models are obtained through different choices of the right-invariant \citep{Diaconis1988} distance $d(\cdot,\cdot)$ on $\mathcal{P}_n$. Right-invariance (see also Definition \ref{rightInv} in the Appendix), which ensures that distances are independent of any relabeling of the items, is an important property in this context, as it ensures that the partition function 
$Z_d(\theta)=
\sum_{\bm{r} \in \mathcal{P}_{n}} e^{-\theta d(\bm{r}, \bm\rho_I)}$
of the MM does not depend on $\bm{\rho}$ \citep{mukherjee2016, vitelli17}.  In the above expression $\bm\rho_I=(1,2,3,\ldots,n)$ denotes the identity permutation.
Nevertheless, the number of terms in the sum makes direct calculation of this partition function unfeasible for all but very small values of $n$. Therefore, the MM is considered known up to a proportionality constant only, except for some particular choices of the distance, for which $Z_d$ may have a closed form \citep{Fligner1986}. Different approximation strategies have been proposed \citep[see e.g.][]{McCullagh1993, mukherjee2016, vitelli17}, allowing inference even with a large number, $n$, of items. Notice that the distance function induces a partition of $\mathcal{P}_n$ formed by sets of rankings which are equidistant from $\bm{\rho}$. Within each partition set, the MM assigns equal probability to all rankings.
As a consequence, exact computation of the partition function is possible for moderate $n$, for some choices of $d$ for which the cardinalities of the partition sets are known \citep[see e.g.][]{irurozki2016permallows, vitelli17}. The partitions of $\mathcal{P}_n$ associated to Spearman's distance play a crucial role in understanding the behavior of the prior proposed here for the MMS. 


In this work we focus on the Mallows model with Spearman's distance, given by $d_S(\bm{r},\bm{\rho})=||\bm{r}-\bm{\rho}||^2 =\sum_{i=1}^n\left({\rho}_i-{r}_i\right)^2,$ for $\bm r, \bm\rho\in\mathcal{P}_n$, which was first introduced with the name rho-model by \citet{Mallows1957}.  Notice that Spearman's distance is an unnormalized version of the Spearman's rank correlation, used to measure the statistical correlation between the ranks of two variables, but, when rankings are considered as vectors in $\mathbb R^n$, it is simply the squared Euclidean distance, or $L_2$-norm. Therefore, we say that a random ranking $\bm R$ follows an MMS distribution, denoted by $\bm R|\bm\rho,\theta\sim\mathcal{M}(\bm\rho,\theta)$, if its probability mass function is given by
 \begin{equation}\label{eq:ProbMassMallows}
  p(\bm{R}\,|\bm{\rho},\theta)\defeq \mathbb{P}(\bm R =\bm{r}\,|\bm{\rho},\theta)= \frac{1}{Z(\theta)}\exp\left[-\theta\, \n{\bm\rho- \bm{R}}\right],
 \end{equation}
where $Z(\theta)\defeq Z_{d_S}(\theta)$ does not have a closed form. Notice, however, that when $\theta = 0$, the MMS reduces to the uniform distribution on $\mathcal{P}_n$.

Given a sample $\bm{R}_1,...,\bm{R}_N|\bm{\rho}, \theta\iidsim\mathcal{M}(\bm{\rho},\theta)$, the likelihood function takes the form 
\begin{equation}\label{likelihood}
\begin{split}
p(\bm{R}_1,...,\bm{R}_N|\bm{\rho},\theta)=\frac{1}{Z(\theta)^N}\exp\left[-\theta\sum_{j=1}^N\n{\bm\rho-\bm R_{j}}\right].
\end{split}
\end{equation}

Therefore, for $\theta>0$, the maximum likelihood estimator (MLE) $\bm\rho_{\text{MLE}}$ is given by
$$\bm{\rho}_{\text{MLE}}=
\argmin_{\bm{\rho}\in\mathcal{P}_n}\sum_{j=1}^N\n{\bm\rho-\bm{R}_j}=
\argmax_{\bm{\rho}\in\mathcal{P}_n}\bm\rho\cdot\bm{\bar R},$$
where the dot denotes the scalar product on $\mathbb{R}^n$, and $\bm{\bar R}=(\bar R_1,\ldots,\bar R_n)$ is the sample mean vector of $\bar R_i=\frac{1}{N}\sum_{j=1}^N R_{ij}$, $i=1,\ldots,n$. 
This is not surprising as the kernel of the MMS distribution coincides with that of an $n$-dimensional Gaussian distribution, except it has a finite support. In other words, the MMS is the restriction of the $n$-dimensional gaussian to $\mathcal P_n$.
Clearly, if $\bm{\bar R}\in\mathcal P_n$, then the MLE simply coincides with the sample mean.  In general, however, $\bm{\bar R}\notin\mathcal P_n$ so a further consideration is required in order to solve the optimization problem. 

\begin{mydef}\label{def:permutohedron}
The \textbf{permutohedron} of order $n$, $\mathbb{pp}_n$, is an $(n-1)$-dimensional polytope embedded in an $n$-dimensional space, the vertices of which are formed by permuting the coordinates of the vector $(1, 2, 3, ..., n)$. Equivalently, it is the convex hull of the points $\bm{\rho}\in\mathcal{P}_n\subset \mathbb{R}^n$.
\end{mydef}

The set $\mathbb{pp}_n$ is sometimes called the \emph{permutation polytope} \citep[see e.g.][]{thompson,Marden1995}. This term, however, refers also to to a similar polytope whose vertices follow a different order. We, here, use the term permutohedron to avoid ambiguity.

\begin{mydef}\label{def:rank_of}
Let $\bm r\in\mathbb{pp}_n$, such that $r_i\neq r_h$ for all $i\neq h$. The rank $\bm{Y}(\bm{r})\in\mathcal{P}_n$ defined by  $Y_i=Y_i(\bm{r})=\sum_{h=1}^n\mathbb{1}(r_h\leq r_i)$, $i=1,...,n$ is called the \textbf{rank vector of} $\bm r$, where $\mathbb{1}(E)$ denotes the indicator function taking the value 1 if the event $E$ is true, and 0 otherwise.
\end{mydef}

By the definition of convex hull, $\bm{\bar R}\in\mathbb{pp}_n$ for any set of rankings, $\bm R_1,\ldots,\bm R_N \in \mathcal{P}_n$. The following proposition shows that $\bm{\rho}_{\text{MLE}}=\bm{Y}(\bar{\bm{R}})$ whenever $\bar R_i\neq \bar R_h$ for all $i\neq h$.
 
\begin{myprop}\label{MIOtheo}
Let $\bm{R}_1,...,\bm{R}_N|\bm{\rho}, \theta\iidsim\mathcal{M}(\bm{\rho},\theta)$, and assume that $\bar{\bm{R}}\in\mathbb{pp}_n$ is such that $\bar{R}_i\not=\bar{R}_h$, for each $i\not=h$. 
Then  $\bm{\rho}_{\text{MLE}}=\bm{Y}(\bar{\bm{R}})$.
\end{myprop}

In order to clarify ideas, consider three samples of size $N=100$ of $5$-dimensional rankings, with sample means $\bm{\bar R}_1 = (1, 3, 4, 5, 2)$, $\bm{\bar R}_2 = (1.1, 3, 4, 4.9, 2)$ and $\bm{\bar R}_3 = (1.5, 2.9, 3.8, 4.7, 2.1)$, respectively. All three sample mean vectors have the same rank vector and, consequently, lead to the same MLE of the consensus ranking, $\bm{\rho}_{\text{MLE}}=\bm{Y}(\bar{\bm{R}}_1)=\bm{Y}(\bar{\bm{R}}_2)=\bm{Y}(\bar{\bm{R}}_3) = \bar{\bm{R}}_1\in\mathcal{P}_5$. Notice, however, that while the rank vector transformation is formally correct, ensuring that the MLE is a proper ranking, it entails a loss of information. Intuitively, looking at the three sample means, one would attach greater uncertainty to the MLE obtained from the third sample, even if this information is lost when looking at the corresponding rank vector, as it is known that a point estimate alone does not provide an uncertainty assessment. 
Definition \ref{def:rank_of} can be generalized to the case of a vector $\bm r$ with ties, i.e., for any $\bm r\in\mathbb{pp}_n$, by letting $Y(\bm r)$ be any ranking whose elements satisfy the same ordering relation as those of $\bm r$. However, the ranking vector in this case would not be unique and so, a unique $\bm\rho_{\text{MLE}}$ would not exist. It would be possible, for instance, to obtain a fourth sample with sample mean $\bar{\bm{R}}_4 = (3, 3, 3, 3, 3)$. Then, any ranking in $\mathcal{P}_5$ would be a rank vector for $\bar{\bm{R}}_4$. This corresponds to a flat likelihood function, for which there is no MLE. Intuitively, any permutation has the same likelihood of being the consensus ranking. This idea is related to the spread or variability of the sample, which in turn is associated to the concentration of the MMS distribution around $\bm\rho$, in other words, to the precision parameter $\theta$. This is, again, not surprising, considering the relation of the MMS with the Gaussian distribution, highlighted above.

It follows that, even if in most applications $\theta$ is considered a nuisance parameter and the main interest is in the estimation of $\bm\rho$, it is nevertheless necessary to estimate $\theta$ in order to get an idea of the reliability of the $\bm\rho$ estimate. A MLE for $\theta$ can be found as the solution to the equation
$$h(\theta)=\frac{Z^\prime(\theta)}{Z(\theta)}-\frac{1}{N}\sum_{j=1}^N\n{\bm{R}_j-\bm{\rho}_\text{MLE}}=0,$$
assuming a unique $\bm{\rho}_\text{MLE}$ exists. This can be done numerically, for instance, via a Newton-Raphson algorithm \citep[see e.g.][]{Marden1995}, but the calculations may be cumbersome for all but small values of $n$.

The Bayesian paradigm is then a natural solution for making inference on the MMS, not only for quantifying uncertainty, but also for including prior information into the statistical analysis. 
In the remainder, we propose and study an informative prior density, specifically tailored to the MMS, building on the Bayesian Mallows model for ranking data of \citet{vitelli17}.

\section{An informative prior}\label{sec:prior}
This section is devoted to the proposal of a prior distribution for the parameters of the MMS. In Section \ref{sub:thetaknown} we analyze the simpler case in which the precision parameter $\theta$ is assumed known. Then, in Section \ref{sub:thetaunknown}, we give an intuition on how to deal with the more general and realistic case of unknown $\theta$.

\subsection{Known precision parameter}\label{sub:thetaknown}
For fixed $\theta$, the likelihood \eqref{likelihood} can be simplified as
\begin{equation}\label{lik1}
\begin{split}
p(\bm{R}_1,...,\bm{R}_N|\bm{\rho},\theta)&\propto \exp\left[2\theta\sum_{j=1}^N\bm\rho\cdot\bm R_{j}\right]\propto\exp\left(2\theta N\bm\rho\cdot\bm{\bar R}\right).
\end{split}
\end{equation}

Therefore, a conjugate prior for $\bm{\rho}\in\mathcal{P}_n$ is given by
\begin{equation}\label{eq:conjPrior}
\pi(\bm{\rho}|\bm{\rho}_0,\eta_0)=\frac{1}{Z^*(\eta_0,\bm{\rho}_0)}\exp\left[-\eta_0||\bm\rho_{0}-\bm\rho||^2\right]\propto \exp\left[2\eta_0\,\bm\rho\cdot\bm\rho_{0}\right].\end{equation}
We call this, the Extended Mallows Model with Spearman distance (EMMS) and write $\bm\rho|\eta_0,\bm\rho_0\sim\mathcal{EM}(\bm{\rho}_0,\eta_0)$. The two parameters $\eta_0\geq 0$ and $\bm{\rho}_0\in\mathbb{pp}_n$ can be interpreted as precision and location parameters, respectively, analogous to those of the MMS. In particular, $\eta_0$ determines the concentration of the distribution around $\bm{\rho}_0$ with $\eta_0=0$ corresponding to a uniform prior on $\mathcal{P}_n$, while larger values reflect stronger prior belief on $\bm{\rho}_0$. 
Notice, however, that the modal parameter cannot be interpreted, in general, as a consensus ranking, except when $\bm{\rho}_0\in \mathcal{P}_n$, in which case the EMMS simply reduces to a MMS. Recall that Mallows models have the limitation that all rankings which are equidistant (in terms of the distance in \eqref{eq:Mallows}) from the consensus ranking have the same probability. For the MMS, this implies in particular that it is not possible to freely assign different masses to different rankings at the same Spearman's distance to the consensus ranking. 
By allowing the modal parameter to take any value in the permutohedron $\mathbb{pp}_n$, that is, to be any convex combination of the elements of  $\mathcal{P}_n$, such structure can be broken, allowing for a more flexible distribution of the mass.  In fact, the prior \eqref{eq:conjPrior} assigns equal mass to all permutations that lie at the same $L_2$-norm from $\bm{\rho}_0$, and greater mass is given to permutations closest to $\bm\rho_0$. For instance, consider the EMMS centered at the barycenter of the permutohedron, that is, with $\bm{\rho}_0= \frac{(n+1)}{2}\bm 1_n$. This results in a uniform distribution on rankings for any value of the precision parameter $\eta_0$. Small deviations from uniformity can be achieved by letting $\eta_0 > 0$ and $||\bm{\rho}_0 - \frac{(n+1)}{2}\bm 1_n||^2$ be small. The direction of the vector $\bm{\rho}_0 - \frac{(n+1)}{2}\bm 1_n$ in $\mathbb{R}^n$ determines the rankings for which the mass increases and those for which it decreases.
   
The case described above, where $\bm{\rho}_0= \frac{(n+1)}{2}\bm 1_n$, is therefore equivalent to assigning to $\bm\rho$ the uniform prior on $\mathcal{P}_n$, $\pi(\bm\rho)=\frac{1}{n!}$, like in \citep{Fligner1990, vitelli17}. As \citet{berger2012objective} discuss, the natural reference prior for a discrete parameter taking values on a finite support is usually the uniform prior on the parameter space. However, the authors show that the uniform prior, in some cases, may not be objective, which may be a property sought by the analyst.
The following result, which holds for the MM with any right-invariant distance, shows that the uniform prior on all rankings, which corresponds to prior \eqref{eq:conjPrior} with (i) $\eta_0=0$ and $\forall \bm\rho_0$, or with (ii) $\bm\rho_0 = \frac{(n+1)}{2}\bm 1_n$ and  $\forall\eta_0$, is formally the objective prior in the sense of \citet{villa2015}. The authors propose that an objective prior for a discrete parameter, in this case the mode of the MMS, should assign to each possible value $\bm\rho$ a mass proportional to the minimum Kullback-Leibler divergence between the model with parameter value $\bm\rho$ and the model with any other parameter value, say $\bm\rho^*\neq\bm\rho$. In this way, prior mass is associated to the ``worth'' of each possible parameter value, defined as the loss in information that would derive from assigning prior probability zero to such value, if it was true.

\begin{myprop}\label{prop:objPrior}
For any right-invariant distance $d$, the objective prior in the sense of \citet{villa2015} for the MM is the uniform prior  on the space of permutations, \normalfont{$\boldsymbol{\rho}\sim\text{Unif}(\mathcal{P}_n)$}.
\end{myprop}


Note that $\bm{\rho}_0\in\mathbb{pp}_n$ may not be a permutation, so that the partition function in \eqref{eq:conjPrior},
\begin{equation}
Z^*(\eta_0,\bm{\rho}_0)=\sum_{\bm{\rho}\in\mathcal{P}_n}\exp\left[-\eta_0||\bm\rho_{0}-{\bm\rho}||^2\right],
\end{equation}
in general depends on $\bm\rho_0$. This implies that \eqref{eq:conjPrior} is known up to a normalization constant. However, in the following sections we show that this drawback can be overcame in practice. Moreover, the fact that $\bm{\rho}_0$ does not need to be a permutation is very convenient for the elicitation problem. The cases where (i) we are interested in including partial information about the consensus ranking, or when (ii) multiple experts' opinions are available, are naturally handled by this prior. 
For instance, imagine a simple example of case (i), where we have some information only on the top-$k$ ranked items. We can define $\bm{\rho}_0$ by fixing the top-$k$ items' ranks $\{1,...,k\}$, and giving the same uniform value $(n+k+1)/2$ to the remaining bottom-$(n-k)$ items' ranks. 
An example of case (ii) is to assume that two (or more) experts believe, \emph{a priori}, in different modal rankings, say $\bm\rho_{01}$ and $\bm\rho_{02}$. An analyst wishing to express an equally strong prior on such two rankings may simply use the prior \eqref{eq:conjPrior} with $\bm\rho_{0}=(\bm\rho_{01}+\bm\rho_{02})/2\in\mathbb{pp}_n$.

A possible reparametrization of \eqref{eq:conjPrior} is obtained by letting $\eta_0=\theta_0 N_0$, which allows to further understand the role of the precision parameter for prior elicitation. One may imagine eliciting prior information from an expert who believes that the consensus ranking is given by some $\bm\rho_0\in\mathcal{P}_n$ and expresses a degree of uncertainty in this belief through some precision, say $\theta_0$. Within scenario (ii), one may imagine that the analyst, having encountered a number $N_0$ of experts who coincide with this view, wishes to summarize this aggregated information by increasing the prior precision. This is achieved by letting $\eta_0=\theta_0 N_0$, thus expressing that individual prior belief regarding $\bm\rho_0$ is reinforced by various experts. 
In the limit, infinite prior precision $\eta_0\rightarrow\infty$ may correspond to a single expert with extremely strong prior belief ($\theta_0\rightarrow\infty$) or to an extremely large number of experts $N_0\rightarrow\infty$ with some prior belief ($\theta_0>0$). Intuitively, one may imagine a situation in which the analyst aggregates prior opinions from many experts by calculating $\rho_0$ and $\eta_0$ as convex and linear combinations, respectively, of the individual $\rho_{0,j}, \theta_{0,j}$ parameters elicited from each expert $j$, and considering the number $N_{0,j}$ of experts who agree on both. In order to understand how this could be done, one may consider how the information passes on from the prior to the posterior. 

The posterior density for $\bm{\rho}$ is given by
\begin{equation}\label{eq:postThetaKnown}
\begin{split}
\pi^N(\bm{\rho}\,|\theta)\propto &\exp\left[2(\eta_0+\theta N)\,\,\bm\rho\cdot \left(\frac{\theta N}{\eta_0+\theta N}\bm{\bar R}+\frac{\eta_0}{\eta_0+\theta N}\bm\rho_{0}\right)\right]
\end{split}
\end{equation}
The first thing we observe is that the proposed prior is indeed conjugate. 
In other words, if $\bm{R}_1,...,\bm{R}_N\,|\,\bm\rho,\theta\iidsim\mathcal{M}(\bm\rho,\theta)$ and $\bm\rho\,|\, \bm\rho_0, \eta_0\sim\mathcal{EM}(\bm\rho_0,\eta_0)$, then it holds  that $\bm{\rho}\,|\,\theta, \bm\rho_0,\eta_0, \bm{R}_1,...,\bm{R}_N\sim \mathcal{EM}(\bm\rho_N,\eta_N)$, with updated parameters:
\begin{align}
&\bm{\rho}_N=\frac{ \theta N}{\eta_0+\theta N}\bar{\bm{R}}+\frac{\eta_0}{\eta_0+\theta N}\bm{\rho}_{0}\in\mathbb{pp}_n\\
&\eta_N= \eta_0 + \theta N \geq 0.
\end{align}
In particular, the prior hyper-parameters $(\bm\rho_0,\eta_0)$ elicited under scenario (ii) can be interpreted as the posterior parameters obtained by an expert who observes $N = \eta_0/\theta$ instances of $\bm\rho_0\in\mathcal{P}_n$, where $\theta$ is the known precision of the true data distribution according to the MMS. Clearly, this interpretation excludes any value of $\eta_0$ that is not a multiple of $\theta$. Nevertheless, such exercise helps to provide an intuition of the role of the prior hyperparameters. Notice that, since any $\bm\rho_0\in\mathbb{pp}_n$ can be expressed as a convex combination of rankings in $\mathcal{P}_n$, the prior mode elicited in scenario (ii) can always be interpreted as arising from multiple (possibly infinite) experts, the calculation of the individually elicited parameters being an exercise in linear algebra. The prior precision parameters for which this interpretation is valid, however, are limited. 
Therefore, in the case in which $\theta$ is assumed known, an interesting case arises by setting $\theta_0=\theta$, i.e. $\eta_0=\theta N_0$.
$N_0$ can be interpreted as an \emph{a priori} sample size, representing the amount of information on which an expert bases the prior belief about the central tendency of $\bm\rho$.
In this sense, the posterior consensus parameter can be viewed as a weighted average of the prior hyper-parameter $\bm{\rho}_{0}$ and the observed mean value $\bar{\bm{R}}$, with weights proportional to the corresponding sample sizes. For any finite prior precision $\theta N_0<\infty$, as the sample size increases, the posterior accumulates mass around $\bm\rho_N$, which approaches the sample mean, $\bm{\bar R}$. Some insights into the role of the prior hyper-parameters can be obtained by considering limiting situations.
An infinite prior precision $N_0=\infty$ would express \emph{a priori} certainty, by accumulating all the prior mass on $\bm\rho_0$, a choice that would make sense only for $\bm\rho_0\in\mathcal{P}_n$. 
The posterior would maintain the infinite precision $\theta_N=\infty$ thus accumulating mass on $\bm\rho_N = \bm\rho_0$. In such hypothetical case, learning would be possible only for infinite sample sizes, with
$$
\lim_{N\rightarrow\infty}\bm{\rho}_N =
\begin{cases}
 \bm{\rho}_{0} & \text{if} \quad N_0/N\rightarrow\infty \\
 (1-\alpha)\bm{\rho}_{0} + \alpha \bar{\bm{R}} & \text{if} \quad N_0/N\rightarrow (1/\alpha-1)\in(0,1)\\
    \bar{\bm{R}} & \text{if} \quad N_0/N\rightarrow 0
   \end{cases}
$$

Notice that, by Proposition \ref{MIOtheo}, the maximum \emph{a posteriori} (MAP) of $\bm{\rho}$ is unique and given by $\bm{\rho}_{\text{MAP}}=\bm{Y}(\bm{\rho}_N)$ provided that all the coordinates of the vector $\bm\rho_N$ take different values. Furthermore, $\eta_N\rightarrow\infty$ as the sample size $N$ grows, thus increasing posterior precision. 


The prior \eqref{eq:conjPrior} has a shape which is analogous to the one discussed earlier by \citet{damien2012}. In their paper, however, the authors propose the use of the Hausdorff distance among subsets (conjugacy classes) of $\mathcal{P}_n$, in place of the squared $L_2$-norm between a ranking and the location parameter of the prior \eqref{eq:conjPrior}, which is an element of the permutation polytope. This difference implies that the proposal of \citet{damien2012} assigns equal probability to all permutations within a conjugacy class. In particular, all rankings in the  modal conjugacy class of the prior are assigned the same mass, even if information may not be available on all such rankings. Furthermore, two permutations in the same class are not necessarily close with respect to the distance used in the MM, which is a crucial element of the model specification. 
Our proposal, instead, is specifically tailored to the MMS, and gives the possibility to choose whether to give maximum prior weight to a unique permutation, or to more than one.  
At the same time, we note in the following Theorem that the results in \citet[][Section 3.3]{damien2012} can be extended to our prior \eqref{eq:conjPrior}.

\begin{theo}\label{teo:Gupta}
Let $D(\bm\rho)=\sum_{j=1}^N\n{\bm R_j - \bm\rho}$, and $D^*(\bm\rho)=\n{\bm\rho_0- \bm\rho}$. Then:
\begin{enumerate}
\item[a)] for each $\bm\rho_1, \bm\rho_2\in\mathcal{P}_n$, and given $\theta, \eta_0$, the ranking $\bm\rho_1$ will have higher posterior probability than $\bm\rho_2$ if and only if
\begin{equation}\label{eq:theo}
    D(\bm\rho_1)-D(\bm\rho_2)<\gamma [D^*(\bm\rho_2)-D^*(\bm\rho_1)],
\end{equation}
where $\gamma = \eta_0/\theta.$
\item[b)] for each $\bm\rho\in\mathcal{P}_n$, if $D^*(\bm\rho)\geq D^*(\bm\rho_\text{MLE})$, $\bm\rho$ will have lower posterior probability than $\bm\rho_\text{MLE}$.
\item[c)] for each $\bm\rho_1, \bm\rho_2\in\mathcal{P}_n$, if $D^*(\bm\rho_1)=D^*(\bm\rho_2)$, $\bm\rho_1$ will have higher posterior probability than $\bm\rho_2$ if and only if $D(\bm\rho_1)<D(\bm\rho_2)$.
\item[d)] for each $\bm\rho_1, \bm\rho_2\in\mathcal{P}_n$, if  $D(\bm\rho_1)<D(\bm\rho_2)$ and $D^*(\bm\rho_1)<D^*(\bm\rho_2)$, then $\bm\rho_1$ will have higher posterior probability than $\bm\rho_2$.
\end{enumerate}
\end{theo}

The theorem, analogous to Gupta and Damien's Theorem 2 and corollaries, gives an intuition on the behavior of the posterior density, by providing a relationship between $\theta$ and $\eta_0$, that determines which rankings receive the highest posterior probabilities. In Section \ref{sec:example} we illustrate, through simulated data, some of the consequences of this theorem on the inference. 

\subsection{Unknown precision parameter}\label{sub:thetaunknown}

When $\theta$ is unknown, the Bayesian paradigm requires a prior on the pair of parameters $(\bm\rho,\theta)$. We here suggest to choose a joint prior of the form $\pi(\bm\rho,\theta) = \pi(\theta)\pi(\bm\rho|\theta)$, where $\pi(\bm\rho|\theta)$ is the EMMS of eq. \eqref{eq:conjPrior}. 
Notice that the particular case of prior independence, $\pi(\bm\rho,\theta) = \pi(\theta)\pi(\bm\rho)$, is achieved in practice by choosing the parameter $\eta_0$ independent of $\theta$. 
Regarding the choice of $\pi(\theta)$ some proposals are present in the literature, for instance an exponential density \citep{vitelli17}, or the conjugate prior of \citet{Fligner1990}.

Alternatively, we suggest the use of the Jeffreys prior for $\theta$, which, in some specific cases, has a closed form and may be an interesting alternative when no information on $\theta$ is available \emph{a priori}. The following proposition holds for any MM with a right-invariant distance, and in particular for the MMS.  

\begin{myprop}\label{prop:jeffrey}
The Jeffreys prior for $\theta$ in a MM with right-invariant distance $d$ takes the form
\begin{equation}\label{eq:jeffrey}
    \pi_J(\theta)=
    \sqrt{ \mathbb{V}_{\bm{R}|\theta}\left[d(\bm{R}, \bm{\rho}_I )|\theta\right]},
\end{equation}
where $\mathbb V_{\bm R|\theta}$ denotes the variance with respect to $\bm R\sim\mathcal M(\bm{\rho}_I,\theta)$, which depends on $\theta$.
\end{myprop}

The posterior density of the model parameters, with the conjugate prior $\pi(\bm\rho|\theta)$  given in eq. \eqref{eq:conjPrior} and possibly independent of $\theta$, and with one of the three prior distributions for $\theta$ mentioned above is 
\begin{equation}\label{eq:post}
\pi^N(\bm\rho,\theta)\propto \frac{\pi(\theta)}{Z^N(\theta)Z^*(\eta_0, \bm\rho_0)}\exp\left\{-\theta N \left[\left(\n{ \bm\rho-\bar{\bm R}}+c_n
- \n{\bar{\bm R}}\right)\right] - \eta_0\n{\bm\rho_0-\bm\rho}\right\}.
\end{equation}

Eq. \eqref{eq:post} can be easily evaluated  in two cases: when (a) $Z^*$ does not depend on $\theta$, that is, when $\eta_0$ is independent of $\theta$, or when (b) $\eta_0= \theta N_0$, and $n$ is small enough, so that $Z^*$ can be calculated exactly, for given prior hyperparameters $\bm\rho_0$ and $N_0$. \\
The more problematic case (c) when $\eta_0= \theta N_0$ and $n$ is too large for computing $Z^*$ exactly, can be handled by using as prior density for $\theta$, $\pi_\text{large n}(\theta)\propto Z^*(\theta N_0, \bm\rho_0)$, so that the posterior density \eqref{eq:post} can be written as 
\begin{equation}\label{eq:post2}
\pi^N(\bm\rho,\theta)\propto \frac{1}{Z^N(\theta)}
\exp\left\{-\theta \left[N\left(\n{ \bm\rho-\bar{\bm R}}+c_n
- \n{\bar{\bm R}}\right) + N_0\n{\bm\rho_0-\bm\rho}\right]\right\}.
\end{equation}

In the next section we sketch the algorithms developed for inference on the MMS in both cases of known and unknown $\theta$, within the situations (a), (b) and (c) described above. 

\section{Posterior simulation}\label{sec:inference}
Notice that, when $\theta=\theta^*$ is known, the posterior \eqref{eq:postThetaKnown} is known up to a normalization constant. Posterior simulation is straightforward in this case and it basically reduces to a visualization problem because of the complexity of the space of permutations. 
In this simple case, we employ a Metropolis-Hastings (MH) Markov Chain Monte Carlo (MCMC) scheme for the update of $\bm{\rho}$.
We propose $\bm{\rho}^{\prime}$ according to the Leap and Shift distribution of \citet{vitelli17}, which is an asymmetric proposal centered around the current value of $\bm{\rho}$. 
We then accept $\bm{\rho}^{\prime}$ with probability $\epsilon=\min\{1,a_{\bm\rho}\}$, where
\begin{equation}\label{eq:accept_rho}
\log a_{\bm\rho}= 2\theta^*(\bm{\rho}^{\prime}-\bm{\rho})\cdot\tilde{\bm R}+\log p_{LS}(\bm{\rho}^{\prime}|\bm{\rho})-\log p_{LS}(\bm{\rho}|\bm{\rho}^{\prime}),
\end{equation}
where, $\tilde{\bm R}=N\bar{\bm R}+N_0\bm\rho_0$, and $p_{LS}$ denotes the transition probability of the Leap and Shift distribution. Notice that, for the sake of simplicity, we are considering the case $\eta_0=\theta^* N_0$, but the results follow trivially for other parametrizations. 

When $\theta$ is not known, we implement a Metropolis within Gibbs scheme for posterior simulation. However, further considerations must be made for the different cases outlined in Section \ref{sub:thetaunknown}. First, we consider case (a), where $\bm\rho$ is assumed \emph{a priori} independent of $\theta$, which amounts to eliciting $\eta_0$ of eq. \eqref{eq:conjPrior} independently of $\theta$; in cases (b) and (c) the precision parameter of the EMMS takes the form $\eta_0 = \theta N_0$.

In (a) $Z^*$ is simply constant, so it creates no additional difficulty. Exact posterior inference can be performed when $n\leq 14$, that is, when we can compute  $Z$ exactly \citep[see][]{vitelli17}. When $n>14$ posterior inference cannot be performed exactly, but we can exploit the efficient scheme of \citet[][Algorithm 1]{vitelli17}, which targets an approximation of the posterior density. Only the acceptance probabilities of the two M-H steps are different here, due to the introduction of the non-uniform prior density on $\bm\rho$. 

In cases (b) and (c) we have the additional issue of dealing with $Z^*$, for which different solutions are possible.
In (b), that is for small $n$, we can compute $Z^*$ on a grid of $\eta_0$ values; whenever its evaluation is required within the M-H step for the update of $\theta$, an approximate value can be obtained via interpolation for values of $\eta_0=\theta N_0$ not in the grid.
In this case we therefore have two steps. First, we update $\bm\rho$ conditional on $\theta$ from the posterior full conditional (see eq.\eqref{eq:post}),
\begin{equation}\label{eq:fullrho}
\pi^N(\bm\rho|\theta)\propto \exp\left[2\theta \bm\rho\cdot \tilde{\bm R}\right].
\end{equation}
This is done as described above, that is, we propose $\bm{\rho}^{\prime}$ according to the Leap and Shift distribution and accept it with probability $\epsilon=\min\{1,a_{\bm\rho}\}$, where $a_{\bm\rho}$ is given in eq. \eqref{eq:accept_rho}, with $\theta^*$ equal to the current value of $\theta$. 
Second, we update $\theta$ conditional on $\bm\rho$. Note that the posterior full conditional for $\theta$ is
\begin{equation}\label{eq:fulltheta}
\pi^N(\theta|\bm\rho)\propto \pi^N(\bm\rho,\theta)\propto \frac{\pi(\theta)}{Z^N(\theta)Z^*(\theta N_0, \bm\rho_0)}\exp\left[-\theta(\tilde{g}-2\bm\rho\cdot \tilde{\bm R})\right],
\end{equation}
where $\tilde{g}=(2N+N_0)c_n+N_0\n{\rho_0}$. The proposal $\theta'$ is sampled from a log-normal density centered on the current value of $\theta$ with a variance tuned in order to obtain a desired acceptance rate. 

In (c), that is, for large values of $n$, only the proposed prior for $\theta$, and therefore  its posterior full conditional, changes and it is given by
\begin{equation}
\pi^N(\theta|\bm\rho)\propto \frac{1}{Z^N(\theta)}\exp\left[-\theta(\tilde{g}-2\bm\rho\cdot \tilde{\bm R})\right].
\end{equation}
Posterior simulation is therefore identical to that of case (b), with the obvious difference in the acceptance probability for $\theta$. 

\section{Illustrative analyses}\label{sec:example}

\subsection{Simulation study}\label{sec:simu1}
In this section we illustrate the effect of the prior on the posterior via a small simulated dataset.
A small $n$ is used so that all possible permutations can be listed.

We generate a sample of $N=30$ rankings from $\mathcal{P}_4$ from the MMS with given \emph{true} parameters $\bm\rho^*=(2,1,4,3)$ and $\theta^*=0.06$. We then set the prior consensus to ${\bm \rho}_{0}=(2, 1, 3,4)$, and perform inference on the model in different settings corresponding to increasing prior sample size for the prior parametrization $\eta_0=\theta N_0$, and the Jeffreys prior for $\theta$. The observed sample mean vector is $\bar{\bm R}=(2.33, 2.17, 3, 2.5)$, which leads to ${\bm \rho}_{\text{MLE}}=\bm Y(\bar{\bm R})=(2, 1, 4, 3)$.
We report in Table \ref{tab:simu} the estimated posterior probability (EPP) of each of the rankings in $\mathcal{P}_4$.
Notice that ${\bm \rho}_{\text{MLE}}$ is the ranking with smallest value of $D(\bm\rho)$ (row highlighted in light-gray and with bold characters). Studying this table, we can verify that Theorem \ref{MIOtheo} holds. 
For instance, solving eq. \eqref{eq:theo} with  $\bm{\rho}_1=\bm{\rho}_0$ and $\bm{\rho}_2=\bm{\rho}_\text{MLE}$, we obtain that $\bm{\rho}_0$ has a higher posterior probability than $\bm{\rho}_\text{MLE}$ if and only if $N_0>15$, which the empirical results confirm.
Also, all rankings $\bm{\rho}$ with $D^*(\bm{\rho})\leq D^*(\bm{\rho}_\text{MLE})$ have lower posterior probabilities than $\bm{\rho}_\text{MLE}$. Furthermore, if $D^*(\bm{\rho}_1)= D^*(\bm{\rho}_2)$, then $\bm{\rho}_1$ has a higher posterior probability than $\bm{\rho}_2$ if $D(\bm{\rho}_1) > D(\bm{\rho}_2)$. 

We can also notice the following \emph{sensitivity} behavior of the posterior probabilities: with increasing $N_0$ the rankings which are closer to $\bm\rho_0$ (in terms of Spearman's distance, or equivalently a smaller $D^*(\bm\rho)$) have increasing posterior probabilities, while those that are farthest, have decreasing posterior probabilities, even if the distance to the data $D(\bm\rho)$ is not so high. An example of this can be seen in the row corresponding to $\bm\rho=(3,1,4,2)$, which has $D(\bm\rho)=230$ and $D^*(\bm\rho)=6$ and for which increasing $N_0$ from 0 to 20 has the effect of decreasing the posterior probability from 0.169 to 0.012. 
The posterior means of $\theta$ in the six settings were 0.068, 0.074, 0.065, 0.06, 0.057, 0.055, while $\theta_\text{MLE} =0.08$ 

\begin{table}[h!]
\centering\footnotesize
\begin{tabular}{|c|c|c|c|c|c|c|c|c|}
  \hline
$\bm\rho$& $D(\bm\rho)$&$D^*(\bm\rho)$& $N_0=0$ & $N_0=5$ & $N_0=10$ & $N_0=15$ & $N_0=16$ &$N_0=20$\\ 
  \hline
  \rowcolor{lightgray}(1,2,3,4) & 260 & 2 & 0.029 & 0.038 & 0.050 & 0.053 & 0.053 & 0.050 \\ 
  (1,2,4,3) & 230 & 4 & 0.172 & 0.125 & 0.080 & 0.052 & 0.050 & 0.036 \\ 
  (1,3,2,4) & 310 & 6 & 0.007 & 0.003 & 0.003 & 0.004 & 0.004 & 0.004 \\ 
  (1,3,4,2) & 250 & 10 & 0.049 & 0.010 & 0.005 & 0.004 & 0.004 & 0.003 \\ 
  (1,4,2,3) & 330 & 12 & 0.004 & 0.001 & 0.001 & 0.002 & 0.002 & 0.001 \\ 
  (1,4,3,2) & 300 & 14 & 0.007 & 0.002 & 0.001 & 0.002 & 0.001 & 0.001 \\ 
  \rowcolor{darkgray}(2,1,3,4) & 250 & 0 & 0.048 & 0.129 & 0.257 & 0.417 & 0.436 & 0.546 \\ 
  \rowcolor{lightgray}
  \textbf{(2,1,4,3)} &\textbf{ 220 }&\textbf{ 2 }&\textbf{ 0.367 }&\textbf{ 0.579 }&\textbf{ 0.527 }&\textbf{ 0.410 }&\textbf{ 0.386}&\textbf{0.303} \\ 
  (2,3,1,4) & 350 & 8 & 0.003 & 0.001 & 0.001 & 0.002 & 0.002 & 0.002 \\ 
  (2,3,4,1) & 260 & 14 & 0.029 & 0.005 & 0.003 & 0.002 & 0.002 & 0.002 \\ 
  (2,4,1,3) & 370 & 14 & 0.002 & 0.001 & 0.001 & 0.001 & 0.001 & 0.001 \\ 
  (2,4,3,1) & 310 & 18 & 0.006 & 0.001 & 0.001 & 0.001 & 0.001 & 0.001 \\ 
  \rowcolor{lightgray}(3,1,2,4) & 290 & 2 & 0.009 & 0.010 & 0.015 & 0.017 & 0.023 & 0.022 \\ 
  (3,1,4,2) & 230 & 6 & 0.169 & 0.065 & 0.032 & 0.016 & 0.017 & 0.012 \\ 
  (3,2,1,4) & 340 & 6 & 0.003 & 0.002 & 0.002 & 0.003 & 0.003 & 0.003 \\ 
  (3,2,4,1) & 250 & 12 & 0.049 & 0.007 & 0.004 & 0.002 & 0.002 & 0.002 \\ 
  (3,4,1,2) & 380 & 18 & 0.002 & 0.001 & 0.001 & 0.001 & 0.001 & 0.001 \\ 
  (3,4,2,1) & 350 & 20 & 0.003 & 0.001 & 0.001 & 0.001 & 0.001 & 0.001 \\ 
  (4,1,2,3) & 300 & 6 & 0.007 & 0.005 & 0.004 & 0.004 & 0.004 & 0.004 \\ 
  (4,1,3,2) & 270 & 8 & 0.019 & 0.007 & 0.006 & 0.004 & 0.004 & 0.004 \\ 
  (4,2,1,3) & 350 & 10 & 0.003 & 0.002 & 0.001 & 0.001 & 0.002 & 0.001 \\ 
  (4,2,3,1) & 290 & 14 & 0.009 & 0.002 & 0.002 & 0.001 & 0.001 & 0.001 \\ 
  (4,3,1,2) & 370 & 16 & 0.002 & 0.001 & 0.001 & 0.001 & 0.001 & 0.001 \\ 
  (4,3,2,1) & 340 & 18 & 0.003 & 0.001 & 0.001 & 0.001 & 0.001 & 0.000 \\ 
   \hline
\end{tabular}
\caption{Results of the simulation study of Section \ref{sec:simu1}. List of the 24 4-rankings (column 1), along with the quantities $D(\bm\rho)$ and $D^*(\bm\rho)$ defined in Theorem \ref{MIOtheo} (columns 2 and 3 respectively). Columns 4 to 9 contain the estimated posterior probabilities of each ranking (rows) and each setting, for increasing values of $N_0$. Four rows are highlighted: in dark-gray, the prior consensus $\bm\rho=\bm\rho_0$ ($D^*(\bm\rho)=0$); in light-gray, the rankings nearest $\bm\rho_0$ ($D^*(\bm\rho)=2$). The  MLE (where $D(\bm\rho)=220$ is minimized) is indicated by bold characters.}\label{tab:simu}
\end{table}

\subsection{\texttt{idea} dataset}
For illustrative purposes, in this section we use the benchmark dataset \texttt{idea} \citep[see e.g.][]{Fligner1990,damien2012}. The data, collected by the Graduate Record Examination (GRE) Board,  consist of a sample of $N=98$ rankings, each of them generated by a college student who was asked to rank $n=5$ words according to their strength of association with the target word `idea'. The five words are `thought' (A), `play' (B), `theory' (C), `dream' (D), and `attention' (E). 
Our aim is to show the effect of our informative prior for $\bm\rho$ on inference. Since $n$ is very small in this example, we can use the exact framework for posterior simulation outlined in Section \ref{sec:inference}, and choose the Jeffreys prior for the parameter $\theta$, thus reflecting our lack of prior knowledge.
In this example, we assume there is reason to believe that $\bm{o}_0$=(A,D,C,B,E) is the true ordering of association of the five words. We therefore choose the corresponding ranking vector $\bm{\rho}_0 =(1, 4, 3, 2, 5)$ as the prior mode. The choice of $N_0$, interpreted as an equivalent sample size, reflects our confidence in $\bm{\rho}_0$, so we consider different settings, corresponding to increasing values of $N_0$. Inference is carried out via MCMC posterior simulation, using a sample size of $5\cdot10^4$ iterations, after a burn-in of $5\cdot10^3$, and the results are shown in Table \ref{tab:idea}.
The orderings corresponding to the most frequently observed rankings in the dataset and their empirical frequencies or sample proportions are shown in columns 1 and 2 respectively, along with their estimated posterior probabilities (EPP) in the different settings (columns 3 to 8). In column 9 we report the Spearman distance between each of the top observed ranking and the prior mode (that is, $D^*(\bm\rho)$).

Recall that our prior  \eqref{eq:conjPrior} assigns equal mass to all rankings at the same Spearman distance from $\bm\rho_0$. This behavior has some analogies with the prior of \citet{damien2012}. However, while there is always a unique ranking at Spearman's distance 0 from $\bm\rho_0$, each conjugacy class contains more than one ranking, all of which are assigned the same mass by the prior of \citet{damien2012}, henceforth GD. As we show below, this difference has a relevant effect on the posterior inferences based on our prior \eqref{eq:conjPrior}, when compared to the results by GD.

\begin{table}[h!]
\centering\footnotesize
\begin{tabular}{|c|c|cccccc|c|}
  \hline
 $\bm{o}$& Prop. & $N_0=0$ & $N_0=1$ & $N_0=5$  & $N_0=10$  & $N_0=49$  & $N_0=98$& $D^*(\bm{\rho})$ \\ 
  \hline
 ACDEB & 0.337 & \textbf{0.047} & \textbf{0.055} & \textbf{0.062} & \textbf{0.076} & 0.105 & 0.078 & 4 \\ 
 ADCEB & 0.184 & 0.037 & 0.044 & 0.050 & 0.060 & \textbf{0.129} & 0.129 & 2\\ 
 ACDBE & 0.122 & 0.031 & 0.035 & 0.041 & 0.048 & 0.095 & 0.103 & 2\\ 
 ADCBE & 0.082 & 0.025 & 0.030 & 0.034 & 0.041 & 0.114 & \textbf{0.176} & 0\\ 
  ACEDB & 0.061 & 0.022 & 0.028 & 0.024 & 0.025 & 0.018 & 0.015 & 10\\ 
 CADEB & 0.051 & 0.025 & 0.028 & 0.030 & 0.026 & 0.023 & 0.019&8 \\
 ADECB & 0.051 & 0.015 & 0.019 & 0.021 & 0.020 & 0.020 & 0.020 &6\\ 
   \hline
\end{tabular}
\caption{Results for the \texttt{idea} dataset. List of orderings corresponding to the rankings with the highest observed frequencies in the data (columns 1 and 2 respectively), along with their EPP in different settings, corresponding to values of $N_0$ between 0 and $N$ (columns 3 to 8). In column 9 we present the Spearman distance between each ranking and the prior mode. The highest EPP of each setting is highlighted in bold characters.}\label{tab:idea}
\end{table}

From this table we can notice the following:
\begin{itemize}
    \item the EPP of (A,D,C,B,E), which corresponds to the prior mode $\bm{\rho}_0$ (row 4), increases consistently with $N_0$; when $N_0=N$, it becomes the posterior modal ranking; 
    \item the ordering (A,C,D,E,B), corresponding to $\bm\rho_{MLE}$ (row 1), remains the ranking with largest EPP provided that the equivalent sample size $N_0$ is not too large. In other words, if the prior does not assign too much mass to $\bm\rho_0\neq\bm\rho_{MLE}$;
    \item the relative ordering of the seven rankings in terms of posterior probability depends on $N_0$, changing for large values which imply strong prior information.
\end{itemize}

Comparing our results with the findings of GD (Table 3), we notice that:
\begin{enumerate}
    \item the posterior distribution of GD places most of the mass (about 0.93) on the top 6 rankings, thus penalizing all other rankings in $\mathcal{P}_5$;
    \item the EPP of the prior modal ranking with ordering (A,D,C,B,E), obtained by GD does not increase with the concentration parameter (in their paper denoted by $\lambda^*$), but rather decreases (from 0.019 when $\lambda^*=0$,  to 0.0067 when $\lambda^*=0.1$). This is not in line with the expected behavior of an informative prior. 
\end{enumerate}

Our posterior distributions, instead, are generally flatter and, importantly, do not show the contradictory behavior with respect to the concentration parameter exhibited by the results of GD and which is probably a consequence of the complex structure of the conjugacy classes of $\mathcal{P}_n$.


\subsection{The prior elicitation problem in practice}\label{ssec:exampleSushi}
In this section we exploit covariates to provide an intuition of how to introduce available information in the prior elicitation problem. 
For the illustration we use the \texttt{sushi} benchmark data of \citet{Kamishima2003}, which consists of full rankings of $n = 10$ different kinds of sushi items given by $N = 5000$ respondents according to their personal preference. The data are available at \url{http://www.kamishima.net/sushi/}. 
This dataset is particularly interesting because it includes covariates of the sushi items. 
We can therefore use this additional information to build an informative prior over the consensus ranking.

We begin from the elicitation of the consensus ranking hyper-parameter $\bm{\rho}_0$ of eq. \eqref{eq:conjPrior}.
We believe that the following covariates of the sushi items (see Table \ref{tab:sushicov}) may have an impact on the personal preference of the respondents: 
\begin{enumerate}
\setlength{\itemsep}{-0.5\baselineskip}
   \item \texttt{oil}: the oiliness in taste (measured on a 0-4 continuous scale, where the smallest the value is, the more oily is the sushi item);
   \item \texttt{eat}: How frequently the sushi item is eaten in sushi shops (measured on a 0-3 continuous scale, where high values correspond to highly frequently sold);
   \item \texttt{price}: the normalized price of the item;
   \item \texttt{sell}: the frequency with which the sushi item is sold  (measured on a 0-1 continuous scale, where high values correspond to highly frequently eaten).
\end{enumerate}

\begin{table}[h!]
\centering\footnotesize
\begin{tabular}{|r|c|c|c|c|}
  \hline
Sushi item & \texttt{oil} & \texttt{eat} & \texttt{price} & \texttt{sell} \\ 
  \hline
shrimp & 2.73 & 2.14 & 1.84 & 0.84 \\ 
  sea eel & 0.93 & 1.99 & 1.99 & 0.88 \\ 
  tuna & 1.77 & 2.35 & 1.87 & 0.88 \\ 
  squid & 2.69 & 2.04 & 1.52 & 0.92 \\ 
  sea urchin & 0.81 & 1.64 & 3.29 & 0.88 \\ 
  salmon roe & 1.26 & 1.98 & 2.70 & 0.88 \\ 
  egg & 2.37 & 1.87 & 1.03 & 0.84 \\ 
  fatty tuna & 0.55 & 2.06 & 4.49 & 0.80 \\ 
  tuna roll & 2.25 & 1.88 & 1.58 & 0.44 \\ 
  cucumber roll & 3.73 & 1.46 & 1.02 & 0.40 \\ 
   \hline
\end{tabular}\caption{Covariate values of interest (columns) for each of the $n=10$ sushi items (rows).}\label{tab:sushicov}
\end{table}

Therefore, we may include the information contained in these four covariates into the analysis of the ranking data, through the following subjective reasoning: the more oily the sushi item is, the more it is preferred; a sushi item which is frequently eaten, is more likely to be preferred than one eaten less frequently; expensive items are preferred above cheaper ones; finally, items which are sold more, are known more and hence preferred. Clearly, the above assumptions are subjective, and someone else may decide to include these covariates differently (for instance, the price may play the opposite role).  
Table \ref{tab:sushirank} shows the rank vectors obtained from these criteria by applying the rank transformation of Definition \ref{def:rank_of} to the covariate vectors of  Table \ref{tab:sushicov}. Notice that the transformation does not result in a proper ranking for the \texttt{sell} variable (column 5): sea eel, tuna, sea urchin and salmon roe have the same covariate value (0.88 in Table  \ref{tab:sushicov}), which results in a tied rank (3.5 in Table \ref{tab:sushirank}). Analogously, shrimp and egg, have the the same value (0.84) resulting a the tied rank (6.5). Nonetheless, the transformed vector for the covariate \texttt{eat} is indeed an element of the permutation polytope $\mathbb{pp}_{10}$, and could therefore be a valid choice for the hyper-parameter $\bm{\rho}_0$. Another interesting feature of Table \ref{tab:sushirank} is that the rankings induced by the different covariates (columns) are not equal but partially agree. A possible choice for the prior consensus hyper-parameter, which takes into account these four different rankings is to set it equal to the average of the rankings induced by the four covariates, that is,
$\bm\rho_{01} = (5.875,3.875,3.625,5.25,4.125,4.125,7.625,3.25,7.25,10)\in\mathbb{pp}_{10}$. Alternatively, one may consider the corresponding rank vector, $\bm\rho_{02}=\bm Y(\bm\rho_{01})=(7,3,2,6,4.5,4.5,9,1,8,10)\in\mathcal{P}_{10}$. 

\begin{table}[h!]
\centering\footnotesize
\begin{tabular}{|r|c|c|c|c|}
  \hline
Sushi item & \texttt{oil} & \texttt{eat} & \texttt{price} & \texttt{sell} \\ 
  \hline
shrimp & 9 & 2 & 6 & 6.5 \\ 
  sea eel & 3 & 5 & 4 & 3.5 \\ 
  tuna & 5 & 1 & 5 & 3.5 \\ 
  squid & 8 & 4 & 8 & 1 \\ 
  sea urchin & 2 & 9 & 2 & 3.5 \\ 
  salmon roe & 4 & 6 & 3 & 3.5 \\ 
  egg & 7 & 8 & 9 & 6.5 \\ 
  fatty tuna & 1 & 3 & 1 & 8 \\ 
  tuna roll & 6 & 7 & 7 & 9 \\ 
  cucumber roll & 10 & 10 & 10 & 10 \\
   \hline
\end{tabular}
\caption{Rank vectors for the sushi items, obtained from the covariates via the rank transformation of Definition \ref{def:rank_of}.}\label{tab:sushirank}
\end{table}

The elicitation of the precision parameter, $\eta_0$, requires a more qualitative reasoning. Considering the parametrization $\eta_0=\theta_0 N_0$, we may decide to fix $N_0=4$, since the consensus hyper-parameter comes from the average of four rankings, which may be interpreted as the opinions of four experts. At the same time, we may choose a relatively large value of $\theta_0$, for instance $\theta_0=0.1$ (which is considered large, given the scale of the problem), thus reflecting confidence in $\bm\rho_0$, given the partial agreement of the four rankings used to construct the consensus hyper-parameter.

\section{Conclusion}\label{sec:conclusion}
In this paper we have proposed an informative prior distribution for the consensus ranking of the Mallows model with Spearman's distance. 
The peculiarity of the proposed prior is that it is a location-scale family for which the location parameter does not need to be a ranking. 
This is convenient for the elicitation problem, since the prior can naturally handle the case when it is difficult to indicate a full ranking which is $\emph{a priori}$ the most likely. For instance, when the total number of items in the application considered is very large, it may be unlikely that an expert is able to elicit a prior ranking over all the items. On the contrary, it may be possible to put some prior information only over the top-ranked items. This is often the case in genomics applications, where thousands of genes are considered in the statistical analysis, but only few of them are known to be related to some disease. Another case which is naturally handled by our prior, is when multiple competing rankings are available prior to the analysis, and we are interested in expressing equally strong prior beliefs on them. 

A limitation, discussed in Section \ref{sec:inference}, arises from the intractability of the normalizing constant $Z^*$ of \eqref{eq:conjPrior} when the location parameter is not itself a ranking. Possible directions for future work include exploring tractable approximations for this quantity, perhaps in the spirit of \citet{mukherjee2016}. In general, more efficient methods for posterior simulations might be developed, but these developments fall outside of the scope of the present work. We do hope, however, that some of the ideas presented here can shed light on potentialities and limitations of the Mallows model with Spearman's distance, and encourage further developments in constructing more flexible priors.

All the simulation algorithms are implemented in \proglang{R} with the \proglang{cpp} package, and will soon be integrated into the \pkg{BayesMallows} \proglang{R} package \citep{BayesMallows}.

\section*{Appendix}\label{app}
Before stating the proof of Proposition \ref{MIOtheo}, let us introduce the formal notion of right-invariance which will prove useful in the proof. 

\begin{mydef}\label{rightInv}{{\bf Right-invariant distance} \citep{Diaconis1988}.}
A distance function is right-invariant, if $d(\bm{\rho},\bm{\sigma})= d(\bm{\rho}\circ\bm{\eta},\bm{\sigma}\circ\bm{\eta})$ for all $\bm{\eta},\bm{\rho},\bm{\sigma}\in\mathcal{P}_n$. With $\bm{\rho}\circ\bm{\eta}$ we denote the composition function of two permutations $\bm{\rho},\bm{\eta}\in\mathcal{P}_n$, which is defined as $\bm{\rho}\circ\bm{\eta}=\bm{\rho}_{\bm{\eta}}=(\rho_{\eta_1},...,\rho_{\eta_n})$.

\end{mydef}

\begin{proof}[Proof of Proposition \ref{MIOtheo}]
The following two identities hold by right-invariance (see Definition \ref{rightInv}):
\begin{align}
&\sum_{i=1}^n\rho_i\bar{R}_i=\sum_{i=1}^ni(\bar{\bm{R}}\circ\bm{\rho}^{-1})_i \label{p1}\\
&\sum_{i=1}^n\rho_i{Y}_i(\bar{\bm{R}})=\sum_{i=1}^ni({Y}(\bar{\bm{R}})\circ\bm{\rho}^{-1})_i=\sum_{i=1}^ni{Y}_i(\bar{\bm{R}}\circ\bm{\rho}^{-1})\label{p2}
\end{align}

Eq. \eqref{p1} implies that $\hat{\bm{\rho}}_1=\argmax_{\bm{\rho}\in\mathcal{P}_n}\sum_{i=1}^ni(\bar{\bm{R}}\circ\bm{\rho}^{-1})_i$ is such that $(\bar{\bm{R}}\circ\hat{\bm{\rho}}_1^{-1})_1\leq (\bar{\bm{R}}\circ\hat{\bm{\rho}}_1^{-1})_2\leq\cdots\leq(\bar{\bm{R}}\circ\hat{\bm{\rho}}_1^{-1})_n$ (by Lemma 2 in  \citet{hullermeier2008label}).\\
By \eqref{p2}, it follows that $\hat{\bm{\rho}}_2=\argmax_{\bm{\rho}\in\mathcal{P}_n}\sum_{i=1}^ni{Y}_i(\bar{\bm{R}}\circ\bm{\rho}^{-1})$, is such that $Y_i(\bar{\bm{R}}\circ\hat{\bm{\rho}}_2^{-1})=i$, for each $i=1,...,n$.\\
Now, notice that  $(\bar{\bm{R}}\circ\hat{\bm{\rho}}_1^{-1})_1\leq (\bar{\bm{R}}\circ\hat{\bm{\rho}}_1^{-1})_2\leq\cdots\leq(\bar{\bm{R}}\circ\hat{\bm{\rho}}_1^{-1})_n$ if and only if $Y_i(\bar{\bm{R}}\circ\hat{\bm{\rho}}_1^{-1})=i$, for each $i=1,...,n$.
This proves that $\hat{\bm{\rho}}_1=\hat{\bm{\rho}}_2$.

\end{proof}

\begin{proof} [Proof of Proposition \ref{prop:objPrior}]
It is sufficient to prove that
 $$\pi(\boldsymbol{\rho}|\theta)\propto \min_{\boldsymbol{\rho}^*\in\mathcal{P}_n\setminus \{\boldsymbol{\rho}\}} KL\big[p(\cdot|\boldsymbol{\rho},\theta)\,||\,p(\cdot|\boldsymbol{\rho}^*,\theta)\big]$$ is independent of $\boldsymbol{\rho}$, for any fixed value of $\theta$, where
 \begin{equation}
     \begin{split}
         KL\big[p(\cdot|&\boldsymbol{\rho},\theta)\,||\,p(\cdot|\boldsymbol{\rho}^*,\theta)\big]= \sum_{\bm r\in\mathcal{P}_n} p(\bm r|\boldsymbol{\rho},\theta)\log\left[\frac{p(\bm r|\boldsymbol{\rho},\theta)}{p(\bm r|\boldsymbol{\rho^*},\theta)}\right]\\ &=\frac{\theta}{Z(\theta)}\sum_{\bm r\in\mathcal{P}_n} \exp\left[-\theta \n{\bm\rho-\bm r}\right]\left[\n{\bm\rho^*-\bm r}-\n{\bm\rho-\bm r}\right]\\
         &=\frac{\theta}{Z(\theta)}\sum_{\bm r\in\mathcal{P}_n} \exp\left[-\theta \n{\bm\rho_I-\bm r}\right]\left[\n{\bm\rho^*-\bm r}-\n{\bm\rho_I-\bm r}\right],
     \end{split}
 \end{equation}
 where the last equality follows from the right-invariance of the Spearman distance. By the same argument, it follows that
 \begin{equation}
     \begin{split}
       \pi(\boldsymbol{\rho}|\theta) & \propto\min_{\boldsymbol{\rho}^*\in\mathcal{P}_n\setminus \{\boldsymbol{\rho}\}} KL\big[p(\cdot|\boldsymbol{\rho},\theta)\,||\,p(\cdot|\boldsymbol{\rho}^*,\theta)\big]\\
       & =\frac{\theta}{Z(\theta)} \min_{\boldsymbol{\rho}^*\in\mathcal{P}_n\setminus \{\boldsymbol{\rho_I}\}}\sum_{\bm r\in\mathcal{P}_n} \exp\left[-\theta \n{\bm\rho_I-\bm r}\right]\left[\n{\bm\rho^*-\bm r}-\n{\bm\rho_I-\bm r}\right],  
     \end{split}
 \end{equation}
does not depend on $\bm\rho$, thus ending the proof.
\end{proof}

\begin{proof}[Proof of Theorem \ref{MIOtheo}]
The proof is  analogous  to  that  of  \citet[Theorem 2]{damien2012}.

\end{proof}

\begin{proof}[Proof of Proposition \ref{prop:jeffrey}]
The Jeffreys prior for a parameter $\theta$ is defined as $\pi(\theta)\propto \sqrt{\mathcal{I}(\theta)}$, where $\mathcal{I}(\theta)$ is the Fisher information function of the statistical model
$$\mathcal{I}(\theta)=-\mathbb{E}_{\bm{R}|\theta}\left[\frac{\text{d}^2}{\text{d}\theta^2}\ln p(\bm{R}|\theta,\bm\rho)\Big|\theta\right].$$
Recall that for the MM it holds 
$$\ln p(\bm{R}|\theta,\bm{\rho})=-\theta d(\bm{R},\bm{\rho})- \ln Z_d(\theta),$$ 
where $Z_d(\theta)=\sum_{\bm{r} \in \mathcal{P}_{n}} e^{-\theta d(\bm{r}, \bm{\rho}_I)}$. Let us simplify the notation here and set $Z_d(\theta):=Z(\theta)$.
Then,
$$\frac{\text{d}^2}{\text{d}\theta^2}\ln p(\bm{R}|\theta,\bm\rho)=-\frac{Z''(\theta)Z(\theta)-[Z'(\theta)]^2}{[Z(\theta)]^2},$$
is independent of $\bm\rho$. 
Notice also that
\begin{equation*}\begin{split}
     Z'(\theta) =& - \sum_{\bm{r} \in \mathcal{P}_{n}} d(\bm{r}, \bm{\rho}_I )e^{-\theta d(\bm{r}, \bm{\rho}_I )} = -Z(\theta) \sum_{\bm{r} \in \mathcal{P}_{n}} d(\bm{r}, \bm{\rho}_I )p(\bm{r}|\theta,\bm\rho=\bm{\rho}_I) \\ =&-\mathbb{E}_{\bm{R}|\theta}[d(\bm{R}, \bm{\rho}_I )|\theta]Z(\theta),
\end{split}
\end{equation*}
and
\begin{equation*}\begin{split}
Z''(\theta)&= \sum_{\bm{r} \in \mathcal{P}_{n}} d^2(\bm{r}, \bm{\rho}_I)e^{-\theta d(\bm{r}, \bm{\rho}_I)}= Z(\theta)\sum_{\bm{r} \in \mathcal{P}_{n}} d^2(\bm{r}, \bm{\rho}_I)p(\bm{r}|\theta,\bm\rho=\bm{\rho}_I)\\
&=\mathbb{E}_{\bm{R}|\theta}[d^2(\bm{R}, \bm{\rho}_I )|\theta]Z(\theta).
\end{split}
\end{equation*}

Then,
\begin{equation*}\begin{split}
\frac{\text{d}^2}{\text{d}\theta^2}\ln p(\bm{R}|\theta,\bm\rho)&=-\frac{\mathbb{E}_{\bm{R}|\theta}[d^2(\bm{R}, \bm{\rho}_I )|\theta][Z(\theta)]^2-[\mathbb{E}_{\bm{R}|\theta}[d(\bm{R}, \bm{\rho}_I )|\theta]]^2[Z(\theta)]^2}{[Z(\theta)]^2}\\
&=-\mathbb{V}_{\bm{R}|\theta}[d(\bm{R}, \bm{\rho}_I)|\theta].\end{split}
\end{equation*}

From the previous equations, it finally holds the result:
$$\pi(\theta)=\sqrt{\mathcal{I}(\theta)}=\sqrt{\mathbb{E}_{\bm{R}|\theta}\left[\mathbb{V}_{\bm{R}|\theta}[d(\bm{R}, \bm{\rho}_I)|\theta]\big|\theta\right]}=\sqrt{\mathbb{V}_{\bm{R}|\theta}\left[d(\bm{R}, \bm{\rho}_I )|\theta\right]}.$$

\end{proof}

\bibliographystyle{ba}
\bibliography{references.bib}

\section*{Acknowledgement}
The authors would like to thank Sonia Petrone, Elja Arjas and Arnoldo Frigessi for their insightful comments.
\end{document}